\documentclass[conference]{IEEEtran}

\usepackage{cite}
\usepackage[utf8]{inputenc}
\usepackage[T1]{fontenc}
\usepackage{amsmath,amssymb,amsfonts}
\usepackage{graphicx}
\usepackage{textcomp}
\usepackage{xcolor}
\usepackage{enumerate}
\usepackage{url}
\usepackage[hypertexnames=false]{hyperref}

\newcommand{\tuple}[1]{\langle #1 \rangle}
\newcommand{\floor}[1]{\left \lfloor #1 \right \rfloor }

\newcommand{\nat}{\mathbb{N}}

\usepackage{algorithm}
\usepackage[noend]{algpseudocode}
\floatname{algorithm}{Algorithm}

\algrenewcommand{\algorithmiccomment}[1]{\small \hfill\# #1}

\algblockdefx{MRepeat}{EndRepeat}{\textbf{repeat}}{}
\algnotext{EndRepeat}

\newtheorem{theorem}{Theorem}

\newtheorem{definition}{Definition}

\newtheorem{remark}{Remark}
\newtheorem{example}{Example}

\begin{document}

\title{Reusing Comparator Networks in Pseudo-Boolean Encodings}

\author{
\IEEEauthorblockN{Micha\l~Karpi\'nski}
\IEEEauthorblockA{\textit{Institute of Computer Science}, \textit{University of Wroc\l aw} \\ Joliot-Curie 15, 50-383 Wroc\l aw, Poland \\ karp@cs.uni.wroc.pl}
\and
\IEEEauthorblockN{Marek~Piotr\'ow}
\IEEEauthorblockA{\textit{Institute of Computer Science}, \textit{University of Wroc\l aw} \\ Joliot-Curie 15, 50-383 Wroc\l aw, Poland \\ mpi@cs.uni.wroc.pl}
}


{\bf IEEE copyright notice}

\noindent \textcopyright 2021 IEEE. Personal use of this material is permitted. Permission from IEEE must be obtained for all
other uses, in any current or future media, including reprinting/republishing this material for advertising
or promotional purposes, creating new collective works, for resale or redistribution to servers or lists, or
reuse of any copyrighted component of this work in other works.

Published in 2022 IEEE International Conference on Tools with Artificial Intelligence.

DOI: 10.1109/ICTAI56018.2022.00117

\newpage

\maketitle

\begin{abstract}
  A Pseudo-Boolean (PB) constraint is a linear inequality constraint over
  Boolean literals. One of the popular, efficient ideas used to solve
  PB-problems (a set of PB-constraints) is to translate them to SAT instances (encodings)
  via, for example, sorting networks, then to process those instances using state-of-the-art SAT-solvers.
  In this paper we show an improvement of such technique. By using a variation of a greedy
  set cover algorithm, when adding constraints to our encoding,
  we reuse parts of the already encoded PB-instance in order to decrease
  the size (the number of variables and clauses) of the resulting SAT instance.  
  The experimental evaluation shows that the proposed method
  increases the number of solved instances.
\end{abstract}


\section{Introduction}

  Boolean satisfiability (SAT) problem has been receiving a continuous interest
  in the field of computer science. Many hard decision problems can be reduced to SAT
  and be efficiently solved by recently-developed SAT-solvers. Some of those problems are
  formulated with the help of different high-level constraints, which should be either
  encoded into CNF formulas or solved inside a SAT-solver by a specialized extension.
  One type of such constraint is a Pseudo-Boolean constraint.

  A {\em Pseudo-Boolean constraint} (a PB-constraint, in short) is a linear inequality with integer coefficients,
  where variables are over Boolean domain. More formally, PB-constraints are of the form $a_1x_1 + a_2x_2 + \cdots + a_nx_n \;\# \; k$,
  where $n,k \in \nat$, $\{x_1,\dots,x_n\}$ is a set of propositional literals (that is, variables or their negations),
  $\{a_1,\dots,a_n\}$ is a set of integer coefficients, and $\# \in \{<,\leq,=,\geq,>\}$. PB-constraints are more expressive and
  more compact than clauses to represent some Boolean formulas,
  especially for optimization problems. PB-constraints are used in many real-life applications,
  for example, in logic synthesis, verification 
  and cumulative scheduling.

  One of the most successful techniques used to solve PB-problems is via translation to SAT, and more specifically,
  via encodings based on comparator networks. This method, first introduced by E{\'e}n and S{\"o}rensson \cite{minisatp},
  has been improved lately by Karpi\'nski and Piotr\'ow \cite{pos18,karpinski2019encoding} and implemented
  in the solver called \textsc{UWrMaxSat} \cite{piotrow2020uwrmaxsat}, which took the first place in Weighted Complete Track
  of the MaxSAT Evaluation 2020 competition. It appears that the same core-guided MaxSAT solving technique
  can be applied to PB-problems and it is used as default in the solver. 

  In this paper we further develop the UWrMaxSat solver by introducing a more efficient encoding of PB-constraints.
  We show that similarities in the structure of consecutive comparator networks introduced by the algorithm
  can be exploited, in order to reduce the size of the resulting CNF formula.

\subsection{Related Work}

  One way to solve a PB-constraint is to transform it to a SAT instance (via Binary Decision Diagrams (BDDs),
  adders or sorting networks \cite{minisatp}) and process it using -- increasingly improving --
  state-of-the-art SAT-solvers. Recent research have favored the approach that uses BDDs,
  which is evidenced by several new constructions and optimizations \cite{sakai2015}. The main advantage of
  BDD-based encodings is that the resulting size of the formula is not dependent on the size of
  the coefficients of a PB-constraint. Karpi\'nski and Piotr\'ow \cite{pos18,karpinski2019encoding} showed that
  encodings based on comparator networks can still be very competitive, and later incorporated their encodings
  in an incremental algorithm for solving SAT instances with Pseudo-Boolean goal functions.

  Another approach to solving Pseudo-Boolean instances is via method called {\em cutting planes},
  originally used in solving integer programs, which is done by iteratively refining a feasible set or
  objective function by means of linear inequalities, called {\em cuts}. It has been observed
  that cutting plane inference is stronger than resolution, therefore it is being used as viable alternative to clause-driven approaches.
  PB-solvers like \textsc{Sat4J}, or more recent \textsc{RoundingSat} \cite{elffers2018divide}, implement this technique.
  Lately, Devriendt et al. extended \textsc{RoundingSat} \cite{devriendt2021cutting} with core-guided search,
  and reported that the cutting planes method allows the solver to derive stronger, non-clausal cores,
  which leads to better updates of the solution bounds, meaning the optimal solution can be found faster.

  PB-constraints are gaining an increasing interest in the MaxSAT community, as one way to solve a MaxSAT instance is to
  encode the maximization objective as the PB-constraint, and then translate it -- as mentioned above --
  to an equisatisfiable SAT instance. We then iteratively enforce the current optimization
  result to be larger than the last one, until the optimum is found. Alternatively,
  a PB-problem can be translated into a MaxSAT-problem and solved by MaxSAT specific algorithms. Several successful ideas
  have emerged using this scheme to solve MaxSAT instances. For example, \textsc{QMaxSAT}
  implements encodings based on totalizer networks: the original totalizer sorting network,
  generalized totalizer networks, mixed radix weighted totalizer,
  and modulo totalizer. Paxian et al.  modified
  the Polynomial Watchdog encoding for solving PB-constraints by replacing the static
  watchdog with a dynamic one allowing to adjust the optimization goal, and showed how to apply it to 
  solve MaxSAT instances.

\subsection{Our Contribution}

  In \textsc{MiniSat+} \cite{minisatp}, the authors implement a scheme to decompose the PB-constraint
  into a number of interconnected sorting networks, where each sorter represents an adder of digits in a mixed radix base.
  The solver \textsc{UWrMaxSat} \cite{piotrow2020uwrmaxsat} implements a modified version of this algorithm,
  where (among other improvements) 4-Way Merge Selection Network \cite{karpinski2019encoding} is used as the underlying comparator network.

  In their paper, Ab{\'\i}o et al. \cite{abio2013parametric} note that their future work would be:
  {\em ``to develop encoding techniques for cardinality constraints that do not process constraints
  one-at-a-time but simultaneously, in order to exploit their similarities''}, and that they:
  {\em ``foresee that the flexibility of the approach presented here [...] will open the door
  to sharing the internal networks among the cardinality constraints present in a SAT problem''}.
  Up to our knowledge, the only paper that explores the possibility of reusing parts of cardinality constraint encodings is \cite{ihalainen2021refined}.
  It presents an algorithm, where, in the process of core-guided MaxSAT solving, a set of cardinality constraints
  is encoded more efficiently by finding common parts among totalizers in a greedy fashion.
  The fact that no further mention or research paper has appeared on this topic creates
  an opportunity for improvement in the field of encoding constraints. To this end we
  have modified the \textsc{UWrMaxSat} solver \cite{piotrow2020uwrmaxsat} in the following way.
  When the new network is introduced (to be encoded) in the interconnected sorting network construction,
  we make use of a variation of the known greedy set cover algorithm \cite{vazirani2013approximation} to try to maximize
  the size of the overlap between input sequences of the previous
  networks and the new one. We aggregate the overlapping input sub-sequences from the previous networks
  and the remaining non-overlapping inputs from the new network, then, those sequences are sorted by length
  and used as an input to the modified version of the 4-Way Merge Selection Network, which utilises a novel
  construction called a Multi-way Merging Network.

  We experimentally compare our new solver with state-of-the-art general constraints solver \textsc{NaPS} \cite{sakai2015}
  and a recently developed PB-solver \textsc{RoundingSat} (the core-guided version of Devriendt et al. \cite{devriendt2021cutting}),
  in order to prove that our techniques are good in practice. The set of benchmarks we use come from the Pseudo-Boolean 2016 competition.

\subsection{Structure of the Paper}

  In Section~\ref{sec:prel} we briefly describe our comparator network algorithm with a novel merging 
  procedure, then we explain the Mixed
  Radix Base technique used in \textsc{UWrMaxSat} and we show how it is applied to encode a PB-constraint by
  constructing a series of comparator networks. In Section~\ref{sec:alg} we show how to leverage the similarities
  between consecutive comparator networks introduced to the solver,
  in order to build an efficient PB-solving algorithm on top of \textsc{UWrMaxSat}. We present the
  results of our experiments in Section~\ref{sec:exp}, and we give concluding remarks in
  Section~\ref{sec:conc}.

\section{Preliminaries}\label{sec:prel}

  In this section we describe two key components for encoding a series of PB-constraints using our method:
  a selection network, and a mixed radix base technique. We start with basic definitions.

  \begin{definition}[Boolean sequences]\label{def:literals}
    A Boolean sequence (or simply -- a sequence) of length~$n$, say $\bar{x} = \tuple{x_1, \dots, x_n}$,
    is an element of $X^n$, where $X$ is a set of Boolean literals (i.e., variables or their negations).
    The length of $\bar{x}$ is denoted by $|\bar{x}|$. The number of occurrences of a
    given literal $l$ in $\bar{x}$ is denoted by $|\bar{x}|_l$.

    Let $S=\{\bar{s}_1, \bar{s}_2,\ldots, \bar{s}_m\}$, be a set of sequences of Boolean literals,
    and let $\bar{d}$ also be a sequence of Boolean literals.
    We say that sequences from $S$ are mutually disjoint in $\bar{d}$,
    if and only if for each literal $l$, $\sum_{1 \leq i \leq m} |\bar{s}_i|_l \leq |\bar{d}|_l$.
    Furthermore, we denote the length of $S$, to be $len(S)=\sum_{1 \leq i \leq m}|\bar{s}_i|$.

    We assume a global ordering of the literals of $X$, that is, literals in each defined sequence are arranged
    according to this ordering.
  \end{definition}

  \begin{definition}[binary sequences]\label{def:binary}
    A binary sequence of length $n$, say $\bar{x} = \tuple{x_1, \dots, x_n}$, is an element of $\{0,1\}^n$.
    We say that a sequence $\bar{x} \in \{0,1\}^n$ is {\em sorted} if $x_i \geq x_{i+1}$, $1 \le i < n$.
    A sequence $\bar{x} \in \{0,1\}^n$ is top $k$ sorted, for $k \leq n$, if $\tuple{x_1,\dots,x_k}$ is sorted and
    $x_k \geq x_i$, for each $i>k$. The length of $\bar{x}$ is denoted by $|\bar{x}|$. Given two binary sequences
    $\bar{x} = \tuple{x_1, \dots, x_n}$ and $\bar{y} = \tuple{y_1,\dots, y_m}$ we define {\em concatenation}
    as $\bar{x} :: \bar{y} = \tuple{x_1, \dots,x_n, y_1, \dots, y_m}$.
  \end{definition}

  \begin{remark}
    The reason we use two separate definitions for sequences is as follows. In Subsection \ref{ssec:mwmn} we
    show an improved merging network from our previous work \cite{karpinski2019encoding}.
    Thus, for completeness, we briefly describe the original merger in Subsection \ref{ssec:4oe}.
    It is easier to present and reason about those constructions in terms of oblivious
    sorting algorithms over the domain of binary sequences. For this reason Definition \ref{def:binary}
    is introduced. In the actual PB-solving (MaxSAT-solving etc.), inputs and outputs of constructions,
    which are later transformed to CNF formulas, are sequences of literals, with possible repetitions.
    Therefore, with the exceptions mentioned above, the rest of the paper uses Definition \ref{def:literals}
    for sequences.
  \end{remark}

  \subsection{4-Way Merge Selection Network}\label{ssec:4oe}

  The main tool in our encoding algorithms is a comparator network. Traditionally comparator networks are
  presented as circuits that receive $n$ inputs and permute them using comparators (2-sorters) connected by
  {\em ``wires''}. Each comparator has two inputs and two outputs. The {\em ``lower''} output is the minimum of inputs, and
  {\em ``upper''} one is the maximum. Their standard definitions and properties can be found, for example, in
  \cite{knuth}. The only difference is that we assume that the output of any sorting
  operation or comparator is in a non-increasing order.

  The main building block of our encoding is a direct selection network, which
  is a certain generalization of a comparator. Encoding of the direct selection
  network of order $(n,k)$ with inputs $\tuple{x_1,\dots,x_n}$ and
  outputs $\tuple{y_1,\dots,y_k}$ is the set of clauses
  $\{x_{i_1} \wedge \dots \wedge x_{i_p} \Rightarrow y_p \, : \, 1 \leq p \leq k, 1 \leq i_1 < \dots < i_p \leq n\}$.
  The direct $n$-sorter is a direct selector of order $(n,n)$.

  Sorting/selection networks used to encode constraints are usually constructed using a~divide-and-conquer principle, similar
  to the merge-sort algorithm. The key component of such network is a merging network -- a network that outputs
  sorted binary sequence (or top $k$ sorted sequence) given outputs of recursive calls. In the following definition
  we assume that comparators are functions and comparator networks are composition of
  comparators. This makes the presentation clear.

  \begin{definition}[m-merger]\label{def:oem}
    A comparator network $f^n_{k}$ is an {\em m-merger of order $(n,k)$},
    if for each tuple $T = \tuple{\bar{x}^1, \dots, \bar{x}^m}$,
    where each $\bar{x}^i$ is a top $k$ sorted binary sequence and $n = \sum_{i=1}^{m} |\bar{x}^i|$,
    $f^n_k(T)$ is top $k$ sorted and is a permutation of $\bar{x}^1 :: \cdots :: \bar{x}^m$.
  \end{definition}

  The formal description of the comparator network used in \textsc{UWrMaxSat} -- called 4-Way Merge Selection Network --
  and the proof of correctness can be found in \cite{karpinski2019encoding}.

  \subsection{Multi-way Merging Network}\label{ssec:mwmn}

  Each merging sub-network of the 4-Way Merge Selection Network receives input sequences of almost
  the same length (except for possibly the last input sequence), which is the consequence of splitting
  the input sequence into subsequences of length 5. Thus, input sequences of each "layer" of mergers are
  always sorted in a non-increasing order by length (notice that Definition \ref{def:oem} does not have such constraint).
  From \cite{karpinski2019encoding}, the number of variables/clauses used in the encoding of network
  $oe\_4merge^n_n(\bar{x}^1,\bar{x}^2,\bar{x}^3,\bar{x}^4)$
  is $O(\left(|\bar{x}^1|+|\bar{x}^2|+|\bar{x}^3|+|\bar{x}^4|\right)\log|\bar{x}^1|)$ (assuming $k=n$, for simplicity).
  So merging arbitrary number of sorted sequences of roughly the same size (for which total length is $n$) using 4-Way Mergers
  will result in an encoding consisting of $O(n\log^2 n)$ variables and clauses.
  In Section \ref{sec:alg} we will see a situation, were we would like to merge $m>1$ sorted sequences of varying lengths.
  Therefore, using the same merging strategy as in the 4-Way Merge Selection Network can be very inefficient. For example, if
  one of the input sequences is of length $n/2$ and the others are of length $c$ ($n/2c$ of them, where $c$ is a constant),
  then the longest sequence will participate in about $O(\log n)$ merging steps, each of them consisting of $O(n\log n)$
  variables and clauses. On the other hand, if we were to first merge all the small sequences, and then merge the longest one in,
  then the longest sequence (in the best scenario) would only need to participate once in the merging 
  process. Asymptotically, this does not change much, but in practice we can save a lot of
  variables and clauses used if we process the input sequences with this {\em skip} technique.
  
  \newdimen{\algindent}
  \setlength\algindent{1.2em}          
  \algnewcommand\LeftComment[2]{\hspace{#1\algindent} $\triangleright$ #2 \hfill}
  \algnewcommand\AND{\textbf{\&} }
  \begin{algorithm}[t!]
    \caption{multiWayMerge}\label{alg:mm}
    \begin{algorithmic}[1]
      \Require{A tuple $\tuple{\bar{x}^1, \dots, \bar{x}^m}$, where each $\bar{x}^i$ is a top $k$ sorted 
      binary sequence,
               $|\bar{x}^i| \geq |\bar{x}^{i+1}|$ (for $1 \leq i < m$), and $n = \sum_{i=1}^{m} 
               |\bar{x}^i|$.}
      \Ensure{The output is top $k$ sorted and is a permutation of the inputs.}
      \State $T_1 \gets \tuple{\bar{x}^1_1, \dots, \bar{x}^m_1} = \tuple{\bar{x}^1, \dots, \bar{x}^m}$; 
      $l=1$
      \State $\overline{out} \gets \emptyset$
      \While {$|T_l| > 1$}
        \State $J_l \gets \{j\leq |T_l|-4: \min(|\bar{x}^j_l|, k) > \sum_{c=1}^{4} |\bar{x}^{j+c}_l|\}$
        \State  $s \gets 0$
        \If{$|T_l| > 4$ \AND $J_l \neq \emptyset$} $s = \max(J_l)$ \EndIf
        \State $s' \gets m - ((m-s)\mod 4)$
        \Statex \LeftComment{1}{Do a merging pass with the 4-Way Merging}
        \For{$j \gets s+1$ \textbf{to} s' \textbf{step} 4}
          \State $oe\_4merge(\bar{x}^{j}_l,\bar{x}^{j+1}_l,\bar{x}^{j+2}_l,\bar{x}^{j+3}_l)$
        \EndFor
        \If{s' < m} $oe\_4merge(\bar{x}^{s'+1}_l,\dots,\bar{x}^{m}_l)$
        \EndIf
        \State Take sequences of indexes $\leq s$ and top $k$ elements from each sequence of the output 
        of
        (8--10) and rename them to $T_l=\{x^1_{l+1},\dots,x^{m'}_{l+1}\}$, in the same order.
        \State Aggregate the non-top $k$ elements and concatenate them with $\overline{out}$.
        \State $l \gets l+1, m \gets m'$
      \EndWhile
      \State \Return $x^1_l :: \overline{out}$      
    \end{algorithmic}
  \end{algorithm}

  We propose a Multi-way Merging Network, where a certain number of the longest input sequences (sorted 
  by length in a non-increasing order)
  are skipped, and the rest of the input sequences are selected for merging in a given iteration,
  only if the resulting sequences do not break the ordering (by length).
  The pseudo-code for this procedure can be seen in Algorithm~\ref{alg:mm}.

  \begin{theorem}
    The output of Algorithm \ref{alg:mm} is top $k$ sorted.
  \end{theorem}
  
  \begin{IEEEproof}
    See extended version of the paper \cite{karpinski2022reusing}.
  \end{IEEEproof}

  \subsection{Mixed Radix Base Technique}\label{ssec:mix}
  The authors of \textsc{MiniSat+} devised a method to decompose
  a PB-constraint into a number of interconnected sorting networks, where sorters
  play the role of adders on unary numbers in a {\em mixed radix representation}.
  Here we present a slightly optimized version proposed in \textsc{NaPS} \cite{sakai2015},
  which is already implemented in \textsc{UWrMaxSat}.

  In the classic base $r$ radix system, positive integers are represented
  as finite sequences of digits $\mathbf{d} = \tuple{d_0,\dots,d_{m-1}}$
  where for each digit $0 \leq d_i < r$, and for the most significant digit, $d_{m-1} > 0$.
  The integer value associated with $\mathbf{d}$ is $v=d_0 + d_1r + d_2r^2+\dots + d_{m-1}r^{m-1}$.
  A mixed radix system is a generalization where a base $\mathbf{B}$ is a
  sequence of positive integers $\tuple{r_0,\dots,r_{m-1}}$. The integer value associated
  with $\mathbf{d}$ is $v = d_0w_0 + d_1w_1 + d_2w_2 + \dots + d_{m}w_{m}$
  where $w_0=1$ and for $i\geq 0$, $w_{i+1} = w_ir_i$. For example, the number
  $\tuple{2,4,10}_{\mathbf{B}}$ in base $\mathbf{B}=\tuple{3,5}$ is interpreted as
  $2 \times \mathbf{1} + 4 \times \mathbf{3} + 10 \times \mathbf{15} = 164$ (values of $w_i$'s in boldface).

  The decomposition of a PB-constraint into sorting networks is roughly as follows:
  first, find a ``suitable'' finite base $\mathbf{B}$ for the given set of coefficients,
  for example, in \textsc{MiniSat+} the base is chosen so that the sum of
  all the digits of the coefficients written in that base is as small as possible. Then for
  each element $r_i$ of $\mathbf{B}$ construct a sorting network where the inputs
  of the $i$-th sorter will be those digits $\mathbf{d}$ (from the coefficients) where $d_i$
  is non-zero, plus the potential carry digits from the $(i-1)$-th sorter.

  We show a construction of a sorting network system using an example. We present
  a step-by-step process of translating a PB-constraint $\psi = 2x_1+2x_2+2x_3+2x_4+5x_5+18x_6 \leq 22$.
  Let $\mathbf{B}=\tuple{2,3,3}$ be the considered mixed radix base.
  Weights of the digit positions of $\mathbf{B}$ are $\bar{w}=\tuple{1,2,6,18}$. First, we normalize the
  constraint by adding a constant $13=\tuple{1,0,2,0}_{\mathbf{B}}$ to both sides of $\psi$, resulting in
  $\psi' = 2x_1+2x_2+2x_3+2x_4+5x_5+18x_6 + 13 < 36$. Notice that $36=\tuple{0,0,0,2}_{\mathbf{B}}$, having 
  only one non-zero digit in base $\mathbf{B}$ (the reason for this step will be revealed later). 

  Thus, the decomposition of the LHS of $\psi'$ is:
  \[
    \mathbf{1} \cdot (1 + x_5) + \mathbf{2} \cdot (x_1 + x_2 + x_3 + x_4 + 2x_5) + \mathbf{6} \cdot (1 + 1) + \mathbf{18} \cdot (x_6)
  \]  
  Now we construct a series of four sorting networks in order to encode the sums at each digit position of $\bar{w}$.
  Given values for the variables, the sorted outputs from these networks represent unary numbers $d_1$,$d_2$,$d_3$,$d_4$
  such that the LHS of $\psi'$ takes the value $\mathbf{1} \cdot d_1 + \mathbf{2} \cdot d_2 + \mathbf{6} \cdot d_3 + \mathbf{18} \cdot d_4$.

  The final step is to encode the carry operation from each digit position to the next. Observe that if a binary sequence
  $\overline{b} = (b_1,\dots,b_n)$ is sorted then for any $d$ and $r$, where $0 \le d \le n$ and $1 < r  \le n$,
  if $\overline{b}$ contains exactly $d$ ones then  $(b_r, b_{2r}, b_{3r}, \dots)$ contains exactly $\floor{\frac{d}{r}}$ ones,
  that is, the carry value for a digit $r$. To this end, each third output of the second and third network
  is fed into the next network as carry input, and the second output of the first network is fed to the second network.
  The full construction is illustrated in Figure \ref{fig:ex1}.

  \begin{figure}[!t]
    \centering
    \includegraphics[scale=0.90]{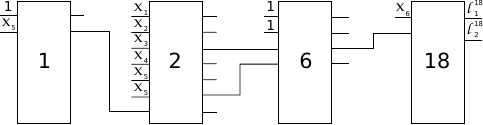}
    \caption{Decomposition of a PB-constraint into a series of interconnected sorting networks.}
    \label{fig:ex1}
  \end{figure}

  To enforce the constraint, we have to add clauses representing the relation $< 36$ (in base $\mathbf{B}$).
  But since we choose the constant $36$ such that in base $\mathbf{B}$ it has only one non-zero digit,
  enforcing the $< 36$ constraint is as easy as adding a singleton clause $\neg l^{18}_2$.
  Notice that outputs of any other network which are not used as carry digits, are irrelevant in enforcing the constraint.

\section{The Algorithm}\label{sec:alg}

  In this section we present our method to reuse comparator networks
  while encoding a series of PB-constraints with mixed radix base technique
  explained in the previous section. We start with an example.

  \begin{example}
  Suppose we want to encode an instance of a PB-problem consisting of (possibly among others) the following 
  two PB-constraints:
    \begin{align*}
      \phi \, = & \,\, x_1 + x_2 + x_3 + x_4 + x_5 < 4, \\
      \psi \, = & \,\, 2x_1 + 2x_2 + 2x_3 + 2x_4 + 5x_5 + 18x_6 + 13 < 36
    \end{align*}

  \noindent We choose the same base as before for $\psi$, i.e., $\mathbf{B}=\tuple{2,3,3}$, but first we encode the constraint~$\phi$.
  Since $\phi$ is a cardinality constraint, its base is empty, and the decomposition consists of a single sorting network (the upper sorter in Figure \ref{fig:reuse}).
  The constraint $\psi$ is the same as in the example from the previous section, but when decomposing it to sorting networks
  we notice, that the input sequence of the second network (which is $\tuple{x_1, x_2, x_3, x_4, x_5, x_5, c}$
  where $c$ is a carry digit; see Figure \ref{fig:ex1}) contains the entire input sequence of the network of the encoding of~$\phi$.
  The fact that we have already sorted the sequence $\tuple{x_1, x_2, x_3, x_4, x_5}$ can now be exploited in the following way.
  We reduce the second sorting network of $\psi$ to contain only two inputs: $x_5$ and $c$.
  We now have two sorted output sequences for inputs $\tuple{x_1, x_2, x_3, x_4, x_5}$ and $\tuple{x_5, c}$.
  Merging them to a single sorted sequence is equivalent to the sorting of the entire
  input sequence $\tuple{x_1, x_2, x_3, x_4, x_5, x_5, c}$, but with fever
  variables and clauses used. To this end, we take the outputs of the first network for the encoding of $\phi$ and feed
  them to the merging network (in this case, a 4-Way Merging Network; see dashed lines in Figure \ref{fig:reuse}).
  We do the same with the outputs of the sorting of $\tuple{x_5, c}$. The rest of the construction is done as in the previous example.
  
  There is another benefit of this technique, which can be demonstrated in this example.
  Notice, that in order to enforce the $< 4$ constraint of $\phi$, we set the fourth
  output of its sorting network to false ($0$, in Figure \ref{fig:reuse}).
  We can also set all the lower outputs to false as well, for better unit propagation.
  Those values will be also fed to any merger which reuses this network,
  as seen in Figure \ref{fig:reuse}. In consequence, those zeros can be used to further simplify mergers
  and any other sorter down the line, as seen in the figure,
  which makes the entire encoding smaller.
  \end{example}

  \begin{figure}[!t]
    \centering
    \includegraphics[scale=0.80]{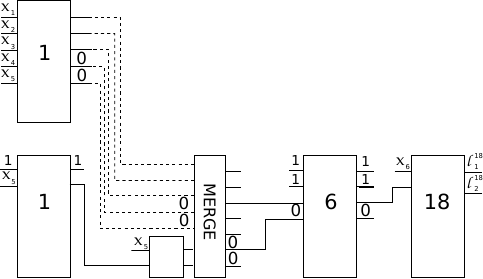}
    \caption{An example of a comparator network being reused in the encoding of another PB-constraint.}
    \label{fig:reuse}
  \end{figure}

  From the example we can see that a simple re-wiring of the inputs/outputs of comparator networks has a potential to reduce the
  size of the encoding of a series of PB-constraints. In contrast to the example, in real-life instances, there might be more than
  one sorter eligible for re-wiring, in fact, there might be more than one way to choose a subset from the available set of sorters
  to be reused in the currently constructed sorter. This leads to the following optimization problem.

  \noindent \textsc{BooleanSequenceCover}\\
  \noindent {\em Input:} A set $S=\{\bar{s}_1,\ldots,\bar{s}_n\}$ (for $n \in \nat$), where each $\bar{s}_i=\tuple{x^i_1,\ldots,x^i_{k_i}}$
                         (for $1 \leq i \leq n$, and $k_i \in \nat$) is a sequence of Boolean literals.
                         Let $\bar{d}=\tuple{x^d_1,\ldots,x^d_m}$ be a sequence of Boolean literals.\\
  \noindent {\em Output:} A set $S' \subseteq S$, where all sequences in $S'$ are mutually disjoint in $\bar{d}$, and the value
                          $len(S')$ is maximized.

  \vspace{\baselineskip}

  \noindent In the \textsc{BooleanSequenceCover} problem, $\bar{d}$ represents an input sequence of the newly constructed sorter,
  and the set $S$ contains input sequences of sorters already present in the encoding.
  The result, a set $S'$, is a set of sorters which inputs cover the most inputs
  of the new sorter. Unfortunately, we do not expect to find a polynomial-time algorithm
  for finding an optimal solution to this problem.

  \begin{theorem}\label{thm:np}
    \textsc{BooleanSequenceCover} is NP-hard.
  \end{theorem}

  \begin{IEEEproof}
    See extended version of the paper \cite{karpinski2022reusing}.
  \end{IEEEproof}

  We settle on an approximate solution, as the consequence of Theorem \ref{thm:np}. To this end, we use a greedy algorithm
  for solving the \textsc{BooleanSequenceCover} problem, presented in Algorithm \ref{alg:1}. It is a simple strategy, where
  we choose a sequence which covers the most elements of $\bar{d}$. Then we remove this sequence from $S$ and $\bar{d}$,
  and the process is repeated until reaching a fix point.

  \begin{algorithm}[t!]
    \caption{greedyBooleanSequenceCover}\label{alg:1}
    \begin{algorithmic}[1]
      \Require{A set $S=\{\bar{s}_1,\ldots,\bar{s}_n\}$ (for $n \in \nat$), where each $\bar{s}_i=\tuple{x^i_1,\ldots,x^i_{k_i}}$
               (for $1 \leq i \leq n$, and $k_i \in \nat$) is a sequence of Boolean literals; a sequence $\bar{d}=\tuple{x^d_1,\ldots,x^d_m}$}
               of Boolean literals.
      \Ensure{A set $S' \subset S$, where all sequences in $S'$ are mutually disjoint in $\bar{d}$.}
      \State $S' \gets \{\}$
      \MRepeat
        \State {\bf let} $\bar{s} \in S$ be the longest subsequence of $\bar{d}$
        \If{$\bar{s} = \emptyset$} \Return $S'$ \EndIf
        \State $S \gets S \setminus \{\bar{s}\}$; $S' \gets S' \cup \{\bar{s}\}$ 
        \State $\bar{d} \gets \bar{d} - \bar{s}$
      \EndRepeat
    \end{algorithmic}
  \end{algorithm}
 
  One can see the similarity to the known greedy set cover algorithm \cite{vazirani2013approximation}
  -- both algorithms pick the set (sequence) that covers the largest number of elements.
  The difference is that here, we try to maximize the number of covered elements,
  while in the set cover problem we minimize the number of sets used in the cover. We note that -- similar to the greedy set cover --
  this strategy does not necessarily output an optimal solution in our case.

  \begin{theorem}\label{thm:main}
    There exist $S,S',O$ and $d$, such that $S',O \subseteq S$, $S'$ is an output of Algorithm~\ref{alg:1} on $S$ and $d$,
    $O$ is an optimal solution to the \textsc{BooleanSequenceCover} problem given inputs $S$ and $d$, such that
    $len(S') = \sqrt{len(O)}$. Moreover, for any $S,S',O$ and $d$, such that $S',O \subseteq S$, $S'$ is 
    an output of Algorithm \ref{alg:1} on $S$ and $d$,
        $O$ is an optimal solution to the \textsc{BooleanSequenceCover} problem given inputs $S$ and $d$,
    we have $len(S') \geq \sqrt{len(O)}$.
  \end{theorem}

  \begin{IEEEproof}
    See extended version of the paper \cite{karpinski2022reusing}.
  \end{IEEEproof}
  
  Even though the technique we use is simple and not optimal, in the next section we show that it is still beneficial in practice.

\section{Experimental Evaluation}\label{sec:exp}

We have implemented the discussed greedy-reuse technique in the PB and MaxSAT solver
\textsc{UWrMaxSat} (commit 8d6c451, compiled without SCIP and MaxPre libraries but with
\textit{bigint} support\footnote{See \url{https://github.com/marekpiotrow/UWrMaxSat}}) and run two
versions of it: \textsc{UWr-R} with greedy-reuse switched on and \textsc{UWr-N} — without it. A
variable \texttt{opt\_reuse\_sorters} in the source code was used to switch between these two versions.
\textsc{COMiniSatPS} by Chanseok Oh \cite{oh2016improving} is our default SAT solver.

The set of instances we have selected as benchmarks is from the most recent Pseudo-Boolean Competition
2016.\footnote{See \url{http://www.cril.univ-artois.fr/PB16/}} We use instances with linear,
Pseudo-Boolean constraints that encode optimization problems, that is, from OPT-SMALLINT-LIN
category (1600 instances).

As hardware, we used the machines with Intel(R) Core(TM) i7-2600 CPU @ 3.40GHz and 16GB of memory
running Ubuntu linux, version 16.04.3. The timeout limit was set to 5000 seconds and the memory
limit is 15 GB, which are enforced with the following commands: \texttt{ulimit -Sv 15872000} and
\texttt{timeout -k 5 5000 <solver> <parameters> <instance>}.

\textsc{UWr-R} found the possibility to reuse some previously-generated sorters in 130 instances
and, in effect, it reduced the number of variables in their encoding by 8.7\% (from 416~659~544 of
\textsc{UWr-N} to 380~573~352) and the number of clauses by 7.7\% (from 1~395~227~515 of
\textsc{UWr-N} to 1~287~723~519) in total.

We have compared our solvers with two state-of-the-art general purpose PB-constraint solvers. 

The first solver is \textsc{NaPS} 1.02b (commit  5956838a, version \textit{bignum}) by Masahiko
Sakai and Hidetomo Nabeshima \cite{sakai2015} which implements improved ROBDD structure for encoding
constraints in band form, as well as other optimizations. As far as we know, it does not contain any
core-guided techniques. This solver was built on the top of \textsc{MiniSat+}, similar to
\textsc{UWrMaxSat}. \textsc{NaPS} won two of the optimization categories in the Pseudo-Boolean
Competition 2016: OPT-BIGINT-LIN and OPT-SMALLINT-LIN. We have launched the main program of
\textsc{NaPS} on each instance, with parameters \texttt{-a -s}.

The second one is \textsc{RoundingSat} (abbreviated to \textsc{RSat}) with core-guided
optimizations (commit b5de84db) by J.~Devriendt, S.~Gocht, E.~Demirović, J.~Nordstr{\"o}m
and P.~J.~Stuckey \cite{devriendt2021cutting}. It extends cutting planes methods with
core-guided search techniques. The results of their experiments presented in the paper
show that it is one of the best PB-solvers nowadays. It was run with just one option
\texttt{--print-sol=1}.

We have launched our solvers \textsc{UWr-R} and \textsc{UWr-N} on each instance, with parameters
\texttt{-a -s -cs}, where \texttt{-cs}  means that in experiments the solver used just one encoding
technique described in this paper: the Multi-way Merge Selection Networks described in Section \ref{sec:prel},
combined with a direct encoding of small sub-networks. Thus, none of the constraints was encoded with BDDs or Adders.

  \begin{figure}[t!]
    \centering
    \includegraphics[width=0.8\columnwidth]{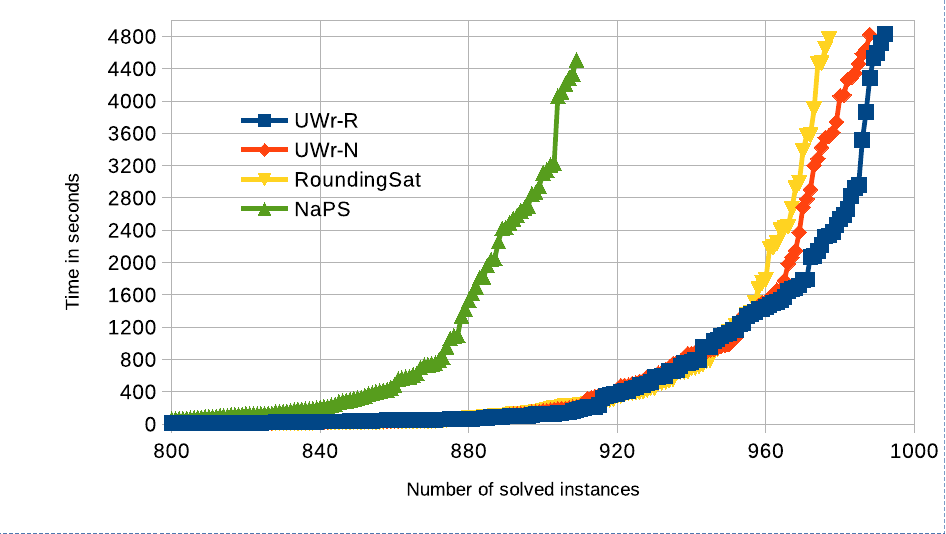}
    \caption{Cactus plot for OPT-SMALLINT-LIN division of PB16 suite}
    \label{fig:plot}
  \end{figure}

  \begin{table}[ht]
    \centering
    \begin{tabular}{ c || c || c | c | c | c | c | }
      solver         & Solved    & Opt       & UnSat     & cpu (h)       & scpu (s)    & avg(scpu)   \\ 
      \hline\hline
\textsc{UWr-R}       & {\bf 992} & {\bf 906} & 86        & {\bf 845.6} & 127 033      & 128.06      \\ 
      \hline
\textsc{UWr-N}       & 988       & 901       & 87        & 858.0       & 135 568      & 137.21      \\ 
      \hline
\textsc{RSat} & 977       & 877       & {\bf 100} & 886.0       & {\bf 97 806} & {\bf 100.11}\\ 
      \hline
\textsc{NaPS}        & 909       & 822       & 87        & 976.1       & 115 977      & 127.59      \\ 
      \hline
      \multicolumn{7}{c}{ } \\
    \end{tabular}
    \caption{Results summary for the OPT-SMALLINT-LIN category}
    \label{tbl:solved}
  \end{table}

In Table \ref{tbl:solved}, we present the numbers of solved instances by each solver. In the {\bf Solved}
column we show the total number of solved instances, which is the sum of the number of instances
where the optimum was found (the {\bf Opt} column) and the number of unsatisfiable instances found
(the {\bf UnSat} column). In the {\bf cpu} column we show the total solving time (in hours) of the
solver over all instances, and {\bf scpu} is the total solving time over solved instances
only. The average has been computed as follows: $\textbf{avg(scpu)} = \textbf{scpu} / \textbf{solved}$.

  \begin{table}[ht]
    \centering
    \begin{tabular}{ c || c | c | c | c | }
      solver  & \textsc{UWr-R} & \textsc{UWr-N} & \textsc{RSat} & \textsc{NaPS} \\ \hline\hline
\textsc{UWr-R}       & -        & 8          & 81            & 92     \\ \hline
\textsc{UWr-N}       & 4        & -          & 80            & 88     \\ \hline
\textsc{RSat} & 66       & 69         & -             & 120    \\ \hline
\textsc{NaPS}        & 9        & 9          & 52            & -      \\ \hline
      \multicolumn{5}{c}{ } \\
    \end{tabular}
    \caption{At the intersection of X row and Y column: the number of instances solved by X but not 
    solved by Y.}
    \label{tbl:symm-diff}
  \end{table}

Results in Table \ref{tbl:solved} and Figure \ref{fig:plot} clearly show that the greedy-reuse
technique in our sorter-base encoding of \textsc{UWr-R} increased the number of solved instances and
reduced the average running time with respect to \textsc{UWr-N}. Therefore, it is chosen as default
in the current version of \textsc{UWrMaxSat}. On the other hand, \textsc{RoundingSat} was the best
in deciding on UNSAT instances and had the best average solving time with respect to solved
instances. Since each of \textsc{UWrMaxSat} and \textsc{RoundingSat} implements different strategies
of solving PB-problems, it is worth noting that there are quite a lot of instances solved by one
solver, but not by the other, and vice versa. See Table \ref{tbl:symm-diff} for details.

\section{Conclusions}\label{sec:conc}

We define the optimization problem \textsc{BooleanSequenceCover} and
show that, despite its NP-hardness, a greedy strategy can be effectively implemented and used to get
smaller sorter-base encoding of industrial instances of NP-hard problems.
The application of an approximation algorithm in PB-constraint encoding techniques gives new
opportunities for many interesting research directions, for example, finding negative
results (i.e., inapproximability results) could be interesting.

Finally, the results of experiments suggest that \textsc{UWrMaxSat} should be probably extended with
some cutting planes methods.

\bibliographystyle{IEEEtran}
\def\IEEEbibitemsep{0pt plus .5pt}
\bibliography{IEEEabrv,reuse}

\begin{thebibliography}{10}
\providecommand{\url}[1]{#1}
\csname url@samestyle\endcsname
\providecommand{\newblock}{\relax}
\providecommand{\bibinfo}[2]{#2}
\providecommand{\BIBentrySTDinterwordspacing}{\spaceskip=0pt\relax}
\providecommand{\BIBentryALTinterwordstretchfactor}{4}
\providecommand{\BIBentryALTinterwordspacing}{\spaceskip=\fontdimen2\font plus
\BIBentryALTinterwordstretchfactor\fontdimen3\font minus
  \fontdimen4\font\relax}
\providecommand{\BIBforeignlanguage}[2]{{%
\expandafter\ifx\csname l@#1\endcsname\relax
\typeout{** WARNING: IEEEtran.bst: No hyphenation pattern has been}%
\typeout{** loaded for the language `#1'. Using the pattern for}%
\typeout{** the default language instead.}%
\else
\language=\csname l@#1\endcsname
\fi
#2}}
\providecommand{\BIBdecl}{\relax}
\BIBdecl

\bibitem{minisatp}
N.~E{\'e}n and N.~S\"orensson, ``Translating pseudo-boolean constraints into
  sat,'' \emph{Journal on Satisfiability, Boolean Modeling and Computation},
  vol.~2, pp. 1--26, 2006.

\bibitem{pos18}
\BIBentryALTinterwordspacing
M.~Karpi{\'n}ski and M.~Piotr{\'o}w, ``Competitive sorter-based encoding of
  pb-constraints into sat,'' in \emph{Proceedings of Pragmatics of SAT 2015 and
  2018}, ser. EPiC Series in Computing, D.~L. Berre and M.~J\"arvisalo, Eds.,
  vol.~59.\hskip 1em plus 0.5em minus 0.4em\relax EasyChair, 2019, pp. 65--78.
  [Online]. Available: \url{https://easychair.org/publications/paper/tsHw}
\BIBentrySTDinterwordspacing

\bibitem{karpinski2019encoding}
------, ``Encoding cardinality constraints using multiway merge selection
  networks,'' \emph{Constraints}, vol.~24, no.~3, pp. 234--251, 2019.

\bibitem{piotrow2020uwrmaxsat}
M.~Piotr{\'o}w, ``Uwrmaxsat: Efficient solver for maxsat and pseudo-boolean
  problems,'' in \emph{2020 IEEE 32nd International Conference on Tools with
  Artificial Intelligence (ICTAI)}.\hskip 1em plus 0.5em minus 0.4em\relax
  IEEE, 2020, pp. 132--136.

\bibitem{sakai2015}
M.~Sakai and H.~Nabeshima, ``Construction of an robdd for a pb-constraint in
  band form and related techniques for pb-solvers,'' \emph{IEICE TRANSACTIONS
  on Information and Systems}, vol.~98, no.~6, pp. 1121--1127, 2015.

\bibitem{elffers2018divide}
J.~Elffers and J.~Nordstr{\"o}m, ``Divide and conquer: Towards faster
  pseudo-boolean solving.'' in \emph{IJCAI}, 2018, pp. 1291--1299.

\bibitem{devriendt2021cutting}
J.~Devriendt, S.~Gocht, E.~Demirovic, J.~Nordstr{\"o}m, and P.~J. Stuckey,
  ``Cutting to the core of pseudo-boolean optimization: combining core-guided
  search with cutting planes reasoning,'' in \emph{Thirty-Fifth AAAI Conference
  on Artificial Intelligence, AAAI 2021, Thirty-Third Conference on Innovative
  Applications of Artificial Intelligence, IAAI 2021, The Eleventh Symposium on
  Educational Advances in Artificial Intelligence, EAAI 2021, Virtual Event},
  2021, pp. 3750--3758.

\bibitem{abio2013parametric}
I.~Ab{\'\i}o, R.~Nieuwenhuis, A.~Oliveras, and E.~Rodr{\'\i}guez-Carbonell, ``A
  parametric approach for smaller and better encodings of cardinality
  constraints,'' in \emph{International Conference on Principles and Practice
  of Constraint Programming}.\hskip 1em plus 0.5em minus 0.4em\relax Springer,
  2013, pp. 80--96.

\bibitem{ihalainen2021refined}
H.~Ihalainen, J.~Berg, and M.~J{\"a}rvisalo, ``Refined core relaxation for
  core-guided maxsat solving,'' in \emph{27th International Conference on
  Principles and Practice of Constraint Programming (CP 2021)}.\hskip 1em plus
  0.5em minus 0.4em\relax Schloss Dagstuhl-Leibniz-Zentrum f{\"u}r Informatik,
  2021.

\bibitem{vazirani2013approximation}
V.~V. Vazirani, \emph{Approximation algorithms}.\hskip 1em plus 0.5em minus
  0.4em\relax Springer Science \& Business Media, 2013.

\bibitem{knuth}
D.~E. Knuth, \emph{The Art of Computer Programming, Volume 3: (2Nd Ed.) Sorting
  and Searching}.\hskip 1em plus 0.5em minus 0.4em\relax Redwood City, CA, USA:
  Addison Wesley Longman Publishing Co., Inc., 1998.

\bibitem{karpinski2022reusing}
M.~Karpi{\'n}ski and M.~Piotr{\'o}w, ``Reusing comparator networks in
  pseudo-boolean encodings,'' \emph{arXiv preprint arXiv:2205.04129}, 2022.

\bibitem{oh2016improving}
C.~Oh, ``Improving sat solvers by exploiting empirical characteristics of
  cdcl,'' Ph.D. dissertation, New York University, 2016.

\end{thebibliography}

\end{document}